\documentclass[11pt,english]{article}
\usepackage[a4paper,bindingoffset=0.0in,left=1in,right=1in,top=0.9in,bottom=0.9in,footskip=.25in]{geometry}

\usepackage{amsfonts,amssymb,amsthm,amsmath,mathrsfs,mathtools}
\usepackage{algorithm,algpseudocode}

\usepackage{bm}
\usepackage{array}
\usepackage[justification=centering,font=footnotesize]{caption}
\usepackage{enumitem,multicol,multirow}
\usepackage{lipsum}
\usepackage{hyperref}
\usepackage{tikz}
\usepackage{fancybox}
\usepackage{textcomp}
\usepackage{fixltx2e}
\usepackage{subfiles}
\usepackage{subcaption}
\usepackage{fancyvrb}
\usepackage{xcolor}
\usepackage{cprotect}
\usepackage{listings}
\usepackage{hhline}
\usepackage{verbatimbox}
\usepackage{framed}
\usepackage{mdframed}
\usepackage{colortbl}
\usepackage{makecell}
\usepackage{sectsty}
\usepackage{textcomp}

\sectionfont{\large}
\subsectionfont{\normalsize}

\newtheorem{theorem}{Theorem}

\newtheorem{proposition}[theorem]{Proposition}

\usepackage{url}
\expandafter\def\expandafter\UrlBreaks\expandafter{\UrlBreaks
  \do\a\do\b\do\c\do\d\do\e\do\f\do\g\do\h\do\i\do\j%
  \do\k\do\l\do\m\do\n\do\o\do\p\do\q\do\r\do\s\do\t%
  \do\u\do\v\do\w\do\x\do\y\do\z\do\A\do\B\do\C\do\D%
  \do\E\do\F\do\G\do\H\do\I\do\J\do\K\do\L\do\M\do\N%
  \do\O\do\P\do\Q\do\R\do\S\do\T\do\U\do\V\do\W\do\X%
  \do\Y\do\Z}

\title{ Can Efficient Detection and Isolation Control an Epidemic? }
\author{
Palash Sarkar \\
Applied Statistics Unit \\
Indian Statistical Institute \\
203, B.~T.~Road \\
Kolkata - 700108 \\
India \\
{\tt palash@isical.ac.in}
}

\begin{document}
\maketitle

\begin{abstract}
The World Health Organisation (WHO) has very strongly recommended testing and isolation as a strategy for controlling the
ongoing COVID-19 pandemic. The goal of this paper is to quantify the effects of detection and isolation using formal models
of epidemics of varying complexity. A key parameter of such models is the basic reproduction ratio. We show that an effective
detection and isolation strategy leads to a reduction of the basic reproduction ratio and can even lead to this ratio becoming
lower than one.  
\end{abstract}

\section{Introduction\label{sec-intro}}
In the context of the ongoing COVID-19 pandemic, on 16th March, 2020, the Director-General (DG) of the World Health Organisation (WHO) gave 
out the following message~\cite{TAG20}. (The emphasis has been added.)
\begin{quotation}
``... the most effective way to prevent infections and save lives is breaking the chains of transmission and to do that you 
must {\em test and isolate}. You cannot fight a fire blindfolded and we cannot stop this pandemic if we don't know who is infected. 
We have a simple message for all countries; test, test, test.''
\end{quotation}
Ideally, if every infected individual is immediately detected and effectively isolated then there can be no transmission of the disease. 
The few initial infections will not spread and the disease will be eradicated. The ideal, however, is not achieved in practice.
On the other hand, the effectiveness of detection and isolation has been established by South Korea where an early large-scale testing
and isolation programme led to a control of the disease.

The goal of this paper is to investigate whether the effects of detection and isolation on controlling an epidemic can be 
quantifiably established. 
The effect of detection is captured by a parameter which we call the detection rate which is the rate
at which infected individuals are detected to have the disease. The individuals who are so detected are isolated/quarantined.
This lowers the transmission rate of the disease for such individuals. Compared to infected individuals who have not been detected,
the detected individuals are likely to have an improved recovery rate due to the access to better supportive treatment. So, the
effects of isolation are possible lowering of the transmission rate and an increase of the recovery rate.

An important parameter in mathematical models of an epidemic is the basic reproduction ratio.
We consider the basic reproduction ratio obtained by the next generation matrix method~\cite{DHJ90,vW02}
at the canonical disease free equilibrium (DFE) where the entire population is susceptible and all other compartments are empty. From 
linear stability analysis it is known~\cite{DHJ90,vW02} that if the basic reproduction ratio is less than $1$, then the 
DFE is asymptotically stable, while if it is greater than $1$, then the DFE is asymptotically unstable, i..e, there is an exponential 
growth in the number of diseased individuals.

We study the effects of detection and isolation on the basic reproduction ratio. To this end, we consider pairs of epidemiological
models. The difference between the models in a pair is that one of the models has a compartment consisting
of all individuals who have been detected to have the infection, while the other model in the pair does not have such a compartment.
Other than this, the two models in a pair have the same compartments and the same parameters. 
For both the models in a pair, we obtain expressions for the basic reproduction ratios using the next generation matrix method
at the canonical disease free equilibrium. 
The first kind of results that we derive show that under natural conditions on the transmission and recovery rate of the detected 
compartment, the basic reproduction ratio
of the model with the detected compartment is less than the basic reproduction ratio of the model without this compartment. This 
demonstrates that effective detection and isolation has a quantifiable effect in lowering of the basic reproduction ratio.
The second kind of results that we prove show that if the detection rate is sufficiently high and for individuals in the
detected compartment either the transmission rate is low enough, or, the recovery rate is sufficiently high,
then the basic reproduction ratio falls below $1$. This quantifiably shows that a sufficiently comprehensive detection and isolation procedure
can actually prevent an epidemic from occurring. 

Two pairs of models are considered. The first pair of models consists of the basic susceptible-infected-recovered (SIR) model and the 
SIDR model, which is the SIR model augmented with a detected compartment. The second pair of models is more complex. The model
without the detected compartment consists of the following compartments: 
susceptible, exposed, infected-asymptomatic, infected-symptomatic and recovered, leading to the SEAIR model. The other
model in the pair is the SEAIDR model consisting of the SEAIR model augmented with a detected compartment. For both pairs of models,
we prove the two kinds of results mentioned above which establish the effects of detection and isolation on the basic reproduction ratio. 

There have been policy guidelines, whereby patients who have been detected to be positive are released from isolation after 
a certain number of days {\em without} testing that they are negative. We model this by augmenting the SEAIDR model with a compartment 
called the pseudo-recovered compartment leading to the SEAIDPR model. Individuals who have been detected to be positive and are released 
after a certain number of days without
testing for being negative are assigned to the pseudo-recovered compartment. Such a policy may have an adverse effect on the basic 
reproduction ratio. The details of the possible adverse effects are worked out in Appendix~\ref{sec-release-wo-test}.

There have been several papers which have studied COVID-19 using compartment models. See for example~\cite{KTLG20,Ber20,SA20,Lietal2020,Ngonghala2020,SNC20,Ma2020,Se2020}.
To the best of our knowledge, no previous work had suggested our approach to quantifying the effects of detection and isolation on the basic
reproduction ratio. 

Before proceeding, we note two points.
\begin{enumerate}
\item 

In the ideal case, the transmission rate of individuals in the detected compartment would be zero. In practice, however, this may 
not hold. For example, there have been many newspaper reports of healthcare workers picking up the infection from patients under treatment. 
For COVID-19, WHO recommends~\cite{WHO-guidelines} that ``that all probable and laboratory-confirmed cases be isolated and cared for in a
health care facility.'' The same document also mentions that ``asymptomatic cases and patients with mild diseases and no risk factors 
can be managed at home, with strict adherence to IPC measures and precautions regarding when to seek care.'' 
The policy guidelines by the Government of India~\cite{MHFW1} allows very mild/pre-symptomatic patients to be under home
quarantine under certain conditions. While both the WHO and the Government of India guidelines provide conditions for home quarantine,
one may note that home quarantine is not necessarily as stringent as institutional quarantine. So, the possibility
of transmission of the disease by people under home quarantine cannot be completely ruled out. Accordingly,
in our models, we do not assume that the transmission rate of individuals in the detected compartment is necessarily zero. While this 
increases the complexity of the calculations, it also allows us to prove conditions on the transmission rate of detected individuals 
for the disease to be controlled. 
\item 

The present work required the computation of expressions for the basic reproduction ratios for the various models. 
In these computations, we were aided by the software SAGE~\cite{sage}. Specifically, initial expressions for the
relevant eigen values of the next generation matrices were obtained using SAGE. The relevant codes are given in
Appendix~\ref{sec-SAGE-code}. For the simpler models, these expressions
were sufficiently simple to be directly handled. On the other hand, for the more complex models, the initial expressions
were quite complicated. As examples, the numerator of the expression for the basic reproduction ratio provided by SAGE 
is a sum of 30 degree-four terms for the SEAIDR model and a sum of 108 degree-five terms for the SEAIDPR model. 
Getting these into tractable forms so that meaningful analysis can be done required a considerable amount of careful analysis.
While the effort behind this analysis is non-trivial, it would be very tedious for a reader to go through the details of 
such algebraic simplifications. So, we have omitted such details and only the final expressions for the basic reproduction ratios are
provided. 
\end{enumerate}


\section{SIDR Model \label{sec-SIR-tested} }
In the SIR model~\cite{KM27} the population is considered to be divided into three disjoint compartments, namely susceptible, infected and
recovered, where those who pass away are included in the recovered compartment. It is assumed that individuals in the recovered
compartment are neither susceptible nor do they infect others. Infected individuals transmit the pathogen to those who are susceptible.
The numbers of susceptible, infected and recovered individuals define the state of the system. The state evolves according to a 
system of differential equations. 
An equilibrium is reached
when there is no further change in the proportions of susceptible, infected and recovered individuals. If the equilibrium is such that
the proportion of infected individuals is equal to zero, then it is called a disease free equilibrium otherwise it is called an
endemic equilibrium. 

We consider the scenario, where members of the population are tested and if found positive are segregated. This is modeled by 
introducing a new compartment in the SIR model which consists of individuals who have been tested and found positive. We call
this the compartment of detected individuals. Individuals in both the infected and detected compartments have the disease.
The difference is that individuals in the detected compartment are known to have the disease while individuals in the infected compartment
are not known to have the disease. This difference changes the behaviour of how the individuals in the detected compartment
are treated. Ideally they would be under isolation/quarantine to ensure that they do not infect others. This signficantly
reduces the transmission rate of individuals in the detected compartment. Also, such individuals are more likely to receive
better supportive treatment and hence are likely to have a higher recovery rate compared to individuals in the infected compartment
who are not known to have the disease.

Moving people from the infected to the detected compartment will be based upon the testing methodology adopted for the population.
A well designed comprehensive testing strategy will ensure that a large number of infected individuals are detected resulting
in the transfer of these individuals from the infected to the detected compartment.

Let $S$, $I$, $D$ and $R$ be numbers of susceptible, infected, detected and recovered individuals respectively.
The quantities $S$, $I$, $D$ and $R$ are functions of time. At any point of time, the invariant $S+I+D+R=N$ holds where $N$ is the 
total size of the population. 

Let $\beta_0$ be the transmission rate for individuals in the infected compartment, which is the product of the contact rate and 
the probability of
transmission given contact. The force of infection for the individuals in the infected compartment is $\lambda_0(I)=\beta_0I/N$. 
Let $\beta_1$ be the transmission rate for individuals in detected compartment and the force of infection for such individuals 
is $\lambda_1(D)=\beta_1D/N$. Individuals in the detected compartment will ideally be under quarantine and separated from susceptible
individuals. So, it is reasonable to expect $\beta_1$ to be lesser than $\beta_0$. In fact, perfect segregation will result in
$\beta_1=0$. 

Let $\gamma_0$ be the rate at which individuals in the infected compartment recover and let $\gamma_1$ be the rate at which individuals
in the detected compartment recover. One may assume that the two rates are equal. There are, however, possible reasons why these two values
may be different. One reason is that when individuals have been found to be positive and moved to the detected compartment, 
they have already spent some time in the infected compartment, so the additional time they need to recover is shorter. Also, when an 
individual has been tested to
be positive, some kind of supportive treatment may be provided which may shorten the recovery time. So, in general, we have
$\gamma_1\geq \gamma_0$ and the condition $\gamma_1>\gamma_0$ may also hold.

Let $\delta$ be the rate at which individuals move from the infected to the detected compartment. The parameter $\delta$ is the detection
rate. The value of $\delta$ is determined by the  efficacy of the testing methodology followed for the entire population. Note that 
$\delta$ does not represent the total number of tests done in the population. Rather it is the rate at which infected individuals are 
detected to be positive. As the total number of tests increases, so does the value of the parameter $\delta$. 

We make the following simplifying assumptions regarding the birth and death rates. Both the birth and death rates are given by 
a parameter $\mu$, all individuals are capable of reproducing and equally subject to mortality and all individuals are born without
infection and are susceptible to infection. These assumptions simplify the mathematical model with respect to demographic
considerations~\cite{BC18}.

It is reasonable to assume that individuals move from the susceptible to the infected compartment. When an infected individual is tested and
found to be positive, the person moves to the detected compartment. Recovery takes place from both the infected and the detected
compartments.
The state of the system is given by the vector $\mathbf{X}=(S,I,D,R)$. Based upon the previous considerations, the evolution of the state with 
respect to time is given by the following system of differential equations.
\begin{eqnarray}\label{eqn-SIDR}
\begin{array}{rcl}
S^{\prime} & = & \mu N -\frac{\beta_0IS}{N} - \frac{\beta_1DS}{N} - \mu S, \\
I^{\prime} & = & \frac{\beta_0IS}{N} + \frac{\beta_1DS}{N} -\delta I - \gamma_0I - \mu I, \\
D^{\prime} & = & \delta I-\gamma_1D - \mu D,\\
R^{\prime} & = & \gamma_0I+\gamma_1D - \mu R.
\end{array}
\end{eqnarray}
Equilibrium in the system given by~(\ref{eqn-SIDR}) is achieved when $S^{\prime}=I^{\prime}=D^{\prime}=R^{\prime}=0$. 
The canonical disease free equilibrium (DFE) is $\mathbf{X}_0=(N,0,0,0)$, i.e., the state where all individuals are susceptible
and the other compartments are empty.

To obtain the basic reproduction ratio, we use the next generation matrix method at the DFE $\mathbf{X}_0$. 
The matrices $\mathbf{F}$ and $\mathbf{V}$ are as follows.
\begin{eqnarray*}
\mathbf{F} & = & 
\left[
\begin{array}{cc}
\beta_0 & \beta_1 \\
0 & 0 
\end{array}
\right], \\
\mathbf{V} & = & 
\left[
\begin{array}{cc}
\delta+\gamma_0+\mu & 0 \\
-\delta & \gamma_1+\mu 
\end{array}
\right].
\end{eqnarray*}
So, the next generation matrix is
\begin{eqnarray*}
\mathbf{F}\mathbf{V}^{-1} & = & 
\frac{1}{(\gamma_1+\mu)(\delta+\gamma_0+\mu)}
\left[
\begin{array}{cc}
\beta_0(\gamma_1+\mu)+\beta_1\delta & \beta_1(\delta+\gamma_0+\mu) \\
0 & 0 
\end{array}
\right].
\end{eqnarray*}
The basic reproduction number $R_0$ is the spectral radius of $\mathbf{F}\mathbf{V}^{-1}$ and is given by the
following expression.
\begin{eqnarray}\label{eqn-R0}
R_0 & = & \frac{\beta_0}{\delta+\gamma_0+\mu} + \frac{\beta_1}{\gamma_1+\mu}\cdot \frac{\delta}{\delta+\gamma_0+\mu}.
\end{eqnarray}


Suppose $\delta=0$. This corresponds to the situation when there is no transfer from the infected to the detected compartment.
Since the only way the detected compartment grows is by transfer from the infected compartment, for $\delta=0$, the detected
compartment is initially empty and stays empty. In other words, there is no detection of positive cases using tests. So, the system evolves 
according to the usual SIR model. Putting $\delta=0$ in~(\ref{eqn-R0}), we see that $R_0=\beta_0/(\gamma_0+\mu)$ which is the
basic reproduction ratio of the SIR model. Let 
\begin{eqnarray}\label{eqn-R0-SIR}
R_0^{\star} & = & \frac{\beta_0}{\gamma_0+\mu}.
\end{eqnarray}
The difference between the SIR model and the SIDR model arises only when $\delta>0$. Under this condition, we have the following
result.
\begin{proposition}\label{prop-lower-SIR}
Suppose $\delta>0$. 
\begin{enumerate}
\item If $\beta_1\leq \beta_0$ and $\gamma_1\geq \gamma_0$, then $R_0\leq R_0^{\star}$.
\item If both the inequalities $\beta_1\leq \beta_0$ and $\gamma_1\geq \gamma_0$ hold and at least one of them is strict, then
$R_0<R_0^{\star}$.
\end{enumerate}
\end{proposition}
\begin{proof}
From~\eqref{eqn-R0} and~\eqref{eqn-R0-SIR}, we have that $R_0\leq R_0^{\star}$ if and only if 
\begin{eqnarray}\label{eqn-R0-R0star}
\frac{\beta_1}{\gamma_1+\mu} & \leq & \frac{\beta_0}{\gamma_0+\mu}
\end{eqnarray}
and $R_0<R_0^{\star}$ if and only if the inequality in~\eqref{eqn-R0-R0star} is strict.

For the first point, note that $\beta_1\leq \beta_0$ and $\gamma_1\geq \gamma_0$ implies~\eqref{eqn-R0-R0star}.
For the second point, note that the given condition implies that the inequality in~\eqref{eqn-R0-R0star} is strict.
\end{proof}

The condition $\delta>0$ corresponds to an effective detection rate. This is achieved only if the testing strategy is effective.
The condition $\beta_1\leq \beta_0$ and $\gamma_1\geq \gamma_0$ with at least one strict inequality corresponds
to the isolation strategy being effective. 
So, whenever we have an effective testing strategy and an effective isolation strategy, the basic reproduction ratio of the
system becomes less than the basic reproduction ratio of the system without any testing. This shows that effective testing and 
isolation has a quantifiable effect in lowering the basic reproduction ratio. 

Suppose $R_0^{\star}>1$. Then in the SIR model without the detected compartment, the DFE is asymptotically unstable and the
disease will have an exponential growth. In contrast, for the SIDR model, we show that there is 
$\delta^{\star}$ such that if $\delta>\delta^{\star}$, then $R_0$ given by~(\ref{eqn-R0}) is less than one. 
\begin{proposition}\label{prop-no-epidemic}
Suppose $\beta_1<\gamma_1+\mu$. Let $\delta^{\star}$ be given by
\begin{eqnarray}\label{eqn-deltastar}
\delta^{\star}
& = & \frac{(\gamma_1+\mu)(\beta_0-\gamma_0-\mu)}{\gamma_1+\mu-\beta_1}.
\end{eqnarray}
Then $R_0<1$ if and only if $\delta>\delta^{\star}$.
\end{proposition}
\begin{proof}
From~(\ref{eqn-R0}), we have $R_0<1$ if and only if 
$$
(\gamma_1+\mu)(\beta_0-\gamma_0-\mu) < \delta (\gamma_1+\mu-\beta_1)
$$
Under the assumption $\beta_1<\gamma_1+\mu$, we have $R_0<1$ if and only if $\delta>\delta^{\star}$.
\end{proof}

The condition $\beta_1<\gamma_1+\mu$ is equivalent to $\beta_1/(\gamma_1+\mu)<1$, i.e., there is no epidemic if we leave out the 
infected people and consider only the individuals in the detected compartment. This is a reasonable assumption, since people who have been
detected to be positive will have a significantly lower transmission rate and the recovery rate will be at least as large as 
those of infected individuals who have not been detected. 

From the expression for $\delta^{\star}$ we see
that $\delta^{\star}>0$ if and only if $R_0^{\star}=\beta_0/(\gamma_0+\mu)>1$. Recall that $R_0^{\star}$ is the basic reproduction ratio
of the system without any detection of infected individuals. If $R_0^{\star}<1$, then there is no epidemic in the
population to start with and the entire analysis becomes vacuous. So, it is meaningful to consider the situation where there
is an epidemic which is equivalent to $R_0^{\star}>1$. In this case, $\delta^{\star}$ is a positive quantity and puts a 
non-trivial lower bound on $\delta$.

The condition $R_0<1$ implies that the SIDR system does not evolve into an epidemic. 
So, we see that if the system would have originally evolved into an epidemic, through effective testing and isolation,
it is possible to ensure that no epidemic occurs. 

\section{A More Complex Model \label{sec-complex} }
We first expand the basic SIR model to include two other compartments, namely exposed and infected-asymptomatic. So, there are a 
total of five compartments, susceptible, exposed, infected-asymptomatic, infected-symptomatic and recovered.
Susceptible individuals move initially to the exposed compartment. Individuals in the exposed compartment have picked up the 
pathogen but are not yet themselves infectious. Individuals move from the exposed compartment to either the
infected-asymptomatic or the infected-symptomatic compartments. People in both the infected-asymptomatic and infected-symptomatic 
compartments are infectious, though, those in the infected-asymptomatic compartment do not yet show the symptoms while those in
the infected-symptomatic compartment exhibit the symptoms of the disease.

At any point of time, let $S$, $E$, $A$, $I$ and $R$ be the numbers of
susceptible, exposed, infected-asymptomatic, infected-symptomatic and recovered individuals respectively. We have $N=S+E+A+I+R$, where
$N$ is the total size of the population. 
The system of differential equations describing the dynamics of the system are given as follows.
\begin{eqnarray} \label{eqn-SEAIR}
\begin{array}{rcl}
S^{\prime} & = & -\frac{\beta_1AS}{N} - \frac{\beta_2IS}{N} + \mu N(0) - \mu S, \\
E^{\prime} & = & \frac{\beta_1AS}{N} + \frac{\beta_2IS}{N} - \sigma_1E-\sigma_2E - \mu E, \\
A^{\prime} & = & \sigma_1E - \gamma_1A - \kappa A -\mu A, \\
I^{\prime} & = & \sigma_2E + \kappa A - \gamma_2I - \mu I, \\
R^{\prime} & = & \gamma_1A+\gamma_2I - \mu R.
\end{array}
\end{eqnarray}
The parameters of the model in~\eqref{eqn-SEAIR} are as follows.
\begin{itemize}
\item $\beta_1$ and $\beta_2$ are the transmission rates of infected-asymptomatic and infected-symptomatic individuals respectively.
\item $\sigma_1$ and $\sigma_2$ are the transfer rates from the exposed compartment to infected-asymptomatic and infected-symptomatic
compartments respectively.
\item $\kappa$ is the transfer rate from infected-asymptomatic compartment to infected-symptomatic compartment.
\item $\gamma_1$ and $\gamma_2$ are the recovery rates of individuals in the infected-asymptomatic and infected-symptomatic
compartments respectively.
\item $\mu$ denotes the birth and death rates.
\end{itemize}
All parameters take non-negative values. In addition, we make the following assumptions.
\begin{itemize}
\item The transmission rates $\beta_1$ and $\beta_2$ are both positive, i.e. individuals in both the
infected-asymptomatic and the infected-symptomatic compartments are infectious.
\item The transfer rate $\sigma_1$ from exposed to infected-asymptomatic is positive, since if $\sigma_1=0$, then there is no
addition to the infected-asymptomatic compartment and the model collapses to SEIR. 
\item At least one of $\sigma_2$ or $\kappa$ is positive, since if both are equal to zero, then there is no addition to the
infected-symptomatic compartment and the model collapses to SEAR.
\end{itemize}
The condition $\beta_1$ and $\beta_2$ are both positive is succintly expressed as $\beta_1\beta_2>0$. Also, the last two
conditions are succintly expressed as $\sigma_1(\sigma_2+\kappa)>0$.

For the stability analysis we consider the DFE $(S,E,A,I,R)=(N,0,0,0,0)$. Applying the next generation matrix method at 
this DFE, we obtain the matrices $\mathbf{F}$ and $\mathbf{V}$ as follows.
\begin{eqnarray*}
\mathbf{F} = 
\left[
\begin{array}{ccc}
0 & \beta_1 & \beta_2 \\
0 & 0 & 0 \\
0 & 0 & 0
\end{array}
\right], 
& & 
\mathbf{V} = 
\left[
\begin{array}{ccc}
\mu + \sigma_1 + \sigma_2 & 0 & 0 \\
-\sigma_1  & \gamma_1 + \kappa + \mu & 0 \\
-\sigma_2  & -\kappa  & \gamma_2 + \mu
\end{array}
\right]
\end{eqnarray*}
The next generation matrix is $\mathbf{F}\mathbf{V}^{-1}$ and the basic reproduction ratio is $R_0^{\star}$ which is given by 
the spectral radius of the next generation matrix. We have
\begin{eqnarray}\label{eqn-R0-star}
R_0^{\star}
& = & 
\frac{\beta_1\sigma_1}{\eta\alpha_1} + \frac{\beta_2\sigma_2}{\eta\alpha_2} + \frac{\beta_2\sigma_1\kappa}{\eta\alpha_1\alpha_2},
\end{eqnarray}
where
\begin{eqnarray}\label{eqn-alpha}
\begin{array}{rcl}
\eta & = & \sigma_1+\sigma_2+\mu, \\
\alpha_1 & = & \gamma_1+\kappa+\mu, \\
\alpha_2 & = & \gamma_2+\mu.
\end{array}
\end{eqnarray}

Let us now consider the introduction of a new compartment in the SEAIR model. This compartment consists of all individuals who have
been detected to be positive giving us the SEAIDR model. In addition to the numbers $S$, $E$, $A$, $I$ and $R$ defined above,
let $D$ be the number of individuals who have been detected to be positive. 
The differential equations describing the dynamics of this model are as follows.
\begin{eqnarray} \label{eqn-SEAIDR}
\begin{array}{rcl}
S^{\prime} & = & -\frac{\beta_1AS}{N} - \frac{\beta_2IS}{N} - \frac{\beta_3DS}{N} + \mu N(0) - \mu S, \\
E^{\prime} & = & \frac{\beta_1AS}{N} + \frac{\beta_2IS}{N} +\frac{\beta_3DS}{N} - \sigma_1E-\sigma_2E -\delta_1E - \mu E, \\
A^{\prime} & = & \sigma_1E - \delta_2A - \gamma_1A - \kappa A -\mu A, \\
I^{\prime} & = & \sigma_2E + \kappa A - \delta_3I - \gamma_2I - \mu I, \\
D^{\prime} & = & \delta_1E+\delta_2A+\delta_3I-\gamma_3D-\mu D, \\
R^{\prime} & = & \gamma_1A+\gamma_2I+\gamma_3D - \mu R.
\end{array}
\end{eqnarray}
In addition to the parameters of the model in~\eqref{eqn-SEAIR}, the following parameters are used in~\eqref{eqn-SEAIDR}.
\begin{itemize}
\item $\beta_3$ is the transmission rate of individuals who have been detected to be positive.
\item $\delta_1,\delta_2,\delta_3$ are the detection rates of individuals in the exposed, infected-asymptomatic and
infected-symptomatic individuals respectively, i.e., these are the rates at which individuals get transferred from 
exposed, infected-asymptomatic and infected-symptomatic compartments to the detected compartment.
\end{itemize}
It is possible that individuals who are exposed but, not yet infectious cannot be detected by a test. In that case $\delta_1=0$.
Individuals who are in the detected compartment are isolated. This lowers the transmission rate of such individuals. In the ideal
case where the isolation is perfect, the transmission rate of individuals in the detected compartment is zero, i.e., $\beta_3=0$. 
On the other hand, considering $\beta_3$ to be positive is more realistic. 

The DFE at which the stability analysis is carried out is $(S,E,A,I,D,R)=(N,0,0,0,0,0)$. The matrices $\mathbf{F}$ and $\mathbf{V}$ 
of the next generation matrix method at this DFE are as follows.
\begin{eqnarray*}
\mathbf{F} & = &
\left[
\begin{array}{cccc}
0 & \beta_1 & \beta_2 & \beta_3 \\
0 & 0 & 0 & 0 \\
0 & 0 & 0 & 0 \\
0 & 0 & 0 & 0
\end{array}
\right], \\
\mathbf{V} & = &
\left[
\begin{array}{cccc}
\sigma_1+\sigma_2+\delta_1+\mu & 0 & 0 & 0 \\
-\sigma_1 & \delta_2+\gamma_1+\kappa+\mu & 0 & 0 \\
-\sigma_2 & -\kappa & \delta_3+\gamma_2+\mu & 0 \\
-\delta_1 & -\delta_2 & -\delta_3 & \gamma_3+\mu
\end{array}
\right].
\end{eqnarray*}
The next generation matrix is $\mathbf{F}\mathbf{V}^{-1}$ and the basic reproduction ratio $R_0$ is the spectral radius of 
$\mathbf{F}\mathbf{V}^{-1}$ which is given by
{\small 
\begin{eqnarray}\label{eqn-R0-main}
R_0 & = & 
\frac{\beta_1\sigma_1}{(\eta+\delta_1)(\alpha_1+\delta_2)}
+\frac{\beta_2\sigma_2}{(\eta+\delta_1)(\alpha_2+\delta_3)}
+\frac{\beta_2\sigma_1\kappa}{(\eta+\delta_1)(\alpha_1+\delta_2)(\alpha_2+\delta_3)} \nonumber \\
& & + \frac{\beta_3}{\gamma_3+\mu}
\left(
\frac{\delta_2\sigma_1+\delta_1(\alpha_1+\delta_2)}{(\eta+\delta_1)(\alpha_1+\delta_2)}
+\frac{\delta_3\sigma_2}{(\eta+\delta_1)(\alpha_2+\delta_3)} 
+\frac{\delta_3\sigma_1\kappa}{(\eta+\delta_1)(\alpha_1+\delta_2)(\alpha_2+\delta_3)}
\right).
\end{eqnarray}
}
If $\delta_1=\delta_2=\delta_3=0$, then $R_0$ given by~\eqref{eqn-R0-main} becomes equal to $R_0^{\star}$ given by~\eqref{eqn-R0-star}.
The condition $\delta_1=\delta_2=\delta_3=0$ signifies that all the three detection rates are zero and so there is no transfer
to the detected compartment. As a result, the system SEAIDR collapses to SEAIR and so the basic reproduction ratio
of the SEAIDR system becomes equal to that of the SEAIR system as mentioned in the previously. The condition that there is a non-trivial
amount of detection is equivalent to at least one of $\delta_1,\delta_2$ and $\delta_3$ being positive which is equivalent to
$\delta_1+\delta_2+\delta_3>0$. We have the following result.

\begin{theorem}\label{thm-less}
Suppose 
\begin{eqnarray}\label{eqn-conds}
\beta_3\leq \beta_2,\ \beta_3\leq \beta_1\left(\frac{\sigma_1}{\sigma_1+\sigma_2+\mu}\right) \mbox{ and } \gamma_3\geq \gamma_1,\gamma_2. 
\end{eqnarray}
Then $R_0\leq R_0^{\star}$.

Further, if $\delta_1+\delta_2+\delta_3>0$, $\sigma_1(\sigma_2+\kappa)>0$, and at least 
one of the inequalities in~(\ref{eqn-conds}) is strict, 
then $R_0<R_0^{\star}$.
\end{theorem}
\begin{proof}
Let $\alpha_3=\gamma_3+\mu$.
From~\eqref{eqn-R0-star} and~\eqref{eqn-R0-main}, we may write
\begin{eqnarray}
\lefteqn{R_0^{\star}-R_0} \nonumber \\
& = & 
\frac{\beta_1\sigma_1}{\eta\alpha_1}\cdot
\frac{(\eta+\delta_1)(\alpha_1+\delta_2)-\eta\alpha_1}{(\eta+\delta_1)(\alpha_1+\delta_2)}
-
\frac{\beta_3}{\alpha_3}\cdot\frac{\delta_2\sigma_1+\delta_1(\alpha_1+\delta_2)}{(\eta+\delta_1)(\alpha_1+\delta_2)} \label{eqn-t00} \\
& & 
+ 
\frac{\beta_2\sigma_2}{\eta\alpha_2} \cdot
\frac{(\eta+\delta_1)(\alpha_2+\delta_3)-\eta\alpha_2}{(\eta+\delta_1)(\alpha_2+\delta_3)}
-
\frac{\beta_3}{\alpha_3}\cdot\frac{\delta_3\sigma_2}{(\eta+\delta_1)(\alpha_2+\delta_3)} \label{eqn-t01} \\
& & 
+\frac{\beta_2\sigma_1\kappa}{\eta\alpha_1\alpha_2}
\cdot\frac{(\eta+\delta_1)(\alpha_1+\delta_2)(\alpha_2+\delta_3)-\eta\alpha_1\alpha_2}
		{(\eta+\delta_1)(\alpha_1+\delta_2)(\alpha_2+\delta_3)}
-\frac{\beta_3}{\alpha_3}\cdot\frac{\delta_3\sigma_1\kappa}{(\eta+\delta_1)(\alpha_1+\delta_2)(\alpha_2+\delta_3)}. \label{eqn-t02}
\end{eqnarray}
From the given conditions, we have $\gamma_3\geq \gamma_1,\gamma_2$ and so 
$\gamma_3+\mu \geq \gamma_1+\mu$ and $\gamma_3+\mu\geq \gamma_2+\mu$ hold, i.e., both $\alpha_3\geq \alpha_1$ and $\alpha_3\geq \alpha_2$ hold.
Also, we have $\beta_3\leq \beta_1\left(\frac{\sigma_1}{\sigma_1+\sigma_2+\mu}\right)$ and $\beta_3\leq \beta_2$.
So, 
\begin{eqnarray}
\frac{\beta_1\sigma_1}{\eta\alpha_1} = \frac{\beta_1\sigma_1}{\alpha_1(\sigma_1+\sigma_2+\mu)}
& \geq & \frac{\beta_3}{\alpha_3}, \label{eqn-t1} \\
\frac{\beta_2}{\alpha_2} & \geq & \frac{\beta_3}{\alpha_3}. \label{eqn-t2}
\end{eqnarray}
Using~(\ref{eqn-t1}) and~(\ref{eqn-t2}), it follows that to show that $R_0^{\star}\geq R_0$, it is sufficient to show the following.
\begin{eqnarray}
(\eta+\delta_1)(\alpha_1+\delta_2)-\eta\alpha_1 & \geq & \delta_2\sigma_1+\delta_1(\alpha_1+\delta_2), \label{eqn-t3} \\
\sigma_2((\eta+\delta_1)(\alpha_2+\delta_3)-\eta\alpha_2) & \geq & \sigma_2\eta\delta_3, \label{eqn-t4} \\
\sigma_1\kappa((\eta+\delta_1)(\alpha_1+\delta_2)(\alpha_2+\delta_3)-\eta\alpha_1\alpha_2) 
& \geq & \sigma_1\kappa\eta\alpha_1\delta_3. \label{eqn-t5}
\end{eqnarray}
Since $\sigma_1\leq \eta$, to establish~(\ref{eqn-t3}), we note the following.
$$
\eta\alpha_1+\delta_2\sigma_1+\delta_1(\alpha_1+\delta_2)
\leq \eta\alpha_1+\delta_2\eta + \delta_1(\alpha_1+\delta_2) = (\eta+\delta_1)(\alpha_1+\delta_2).
$$
To establish~(\ref{eqn-t4}), we note the following.
$$ \eta\alpha_2+\eta\delta_3 = \eta(\alpha_2+\delta_3)\leq (\eta+\delta_1)(\alpha_2+\delta_3).  $$
To establish~(\ref{eqn-t5}), we note the following.
$$ \eta\alpha_1\alpha_2+\eta\alpha_1\delta_3 = \eta\alpha_1(\alpha_2+\delta_3)\leq (\eta+\delta_1)(\alpha_1+\delta_2)(\alpha_2+\delta_3).  $$
This completes the proof of $R_0^{\star}\geq R_0$.

We now consider conditions for $R_0^{\star}$ to be strictly greater than $R_0$.
First suppose that $\beta_3=0$. Then the second terms of the expressions in~(\ref{eqn-t00}),~(\ref{eqn-t01}) and~(\ref{eqn-t02}) are 
all equal to zero.
Under the conditions $\delta_1+\delta_2+\delta_3>0$ and $\sigma_1(\sigma_2+\kappa)>0$, it follows that the first term
of at least one of the expressions in~(\ref{eqn-t00}),~(\ref{eqn-t01}) or~(\ref{eqn-t02}) is positive. So, $R_0^{\star}>R_0$.

Now suppose $\beta_3>0$. If at least one of the inequalities in~(\ref{eqn-conds}) is strict, it follows that one of the
inequalities in~(\ref{eqn-t1}) or~(\ref{eqn-t2}) is strict. Further, under the conditions $\delta_1+\delta_2+\delta_3>0$ 
and $\sigma_1(\sigma_2+\kappa)>0$ one of the inequalities in~(\ref{eqn-t3}),~(\ref{eqn-t4}) or~(\ref{eqn-t5}) is strict.
So, $R_0^{\star}>R_0$.

This completes the proof.
\end{proof}

The conditions in the statement of Theorem~\ref{thm-less} are quite natural. We discuss these in details.
\begin{description}
\item{\em $\gamma_3\geq \gamma_1,\gamma_2$:} This condition
expresses the fact that the recovery rate of individuals in the detected compartment is at least as much as the recovery rates
of individuals in the infected-asymptomatic and the infected-symptomatic compartments. 
\item{\em $\beta_3\leq \beta_2$ and $\beta_3\leq \beta_1\left(\sigma_1/(\sigma_1+\sigma_2+\mu)\right)$:} 
The parameter $\beta_3$ is the transmission rate of individuals in the detected compartment. In the ideal case when the isolation
is perfect, $\beta_3=0$. In practice, the isolation may not be perfect and there is a possibility that an individual in the detected 
compartment transmits the infection. It is, however, reasonable to assume that this transmission rate is small. 
The condition $\beta_3\leq \beta_2$ expresses the fact that the
transmission rate of individuals in the detected compartment is at most the transmission rate of individuals in the
infected-symptomatic compartment. The condition $\beta_3\leq \beta_1\left(\sigma_1/(\sigma_1+\sigma_2+\mu)\right)$ is somewhat 
stronger in that it requires the transmission rate of individuals in the detected compartment to be at most a 
fraction $\left(\sigma_1/(\sigma_1+\sigma_2+\mu)\right)$ of the transmission rate of individuals in the
infected-asymptomatic compartment. If $\beta_3$ is small, then this can be expected to hold.
\item{\em $\delta_1+\delta_2+\delta_3>0$:} This condition is equivalent to saying that at least one of $\delta_1$, $\delta_2$ and $\delta_3$
is positive, i.e., there is a non-zero transfer of individuals from one of exposed, infected-asymptomatic and
infected-symptomatic compartments to the detected compartments. 
\item{\em $\sigma_1(\sigma_2+\kappa)>0$:} As explained earlier, this condition merely states that there are non-empty transfers
into the infected-asymptomatic and infected-symptomatic compartments.
\item{\em One of the inequalities in~(\ref{eqn-conds}) is strict:} 
At least one of the inequalities in~(\ref{eqn-conds}) being strict expresses the fact that either the recovery rate of individuals
in the detected compartment is strictly better than at least one of the recovery rates of individuals in the
infected-asymptomatic or the infected-symptomatic compartments, or, the transmission rate of individuals in the
detected compartment is strictly less than the transmission rate of individuals in the infected-symptomatic compartment, or
it is strictly less than the transmission rate of individuals in the infected-asymptomatic compartment adjusted by a factor.
\end{description}
In summary, the condition $\delta_1+\delta_2+\delta_3>0$ amounts to saying that the detection rate is non-zero and
the condition that one of the inequalities in~(\ref{eqn-conds}) is strict amounts to saying that the isolation and treatment
is effective. So, Theorem~\ref{thm-less} shows that if the overall detection rate is positive and the overall isolation and treatment 
are effective, then the basic reproduction ratio becomes lower than what it otherwise would be.

We now consider the condition under which $R_0$ becomes less than one. The following result provides the limiting value
of $R_0$ as the detection rates $\delta_2$ and $\delta_3$ are increased.
\begin{theorem}\label{thm-asymptotic}
For fixed values of $\beta_1,\beta_2,\beta_3,\gamma_1,\gamma_2,\gamma_3,\sigma_1,\sigma_2,\kappa$ and $\delta_1$, 
$$R_0\rightarrow \frac{\beta_3}{\gamma_3+\mu}\cdot \frac{\sigma_1+\sigma_2+\delta_1}{\mu+\sigma_1+\sigma_2+\delta_1}, \quad \mbox{ as }
\delta_2,\delta_3\rightarrow \infty. $$ 
\end{theorem}
\begin{proof}
Note that as $\delta_2,\delta_3\rightarrow \infty$, the first three terms in the expression for $R_0$ given by~(\ref{eqn-R0-main})
goes to 0. So, we only need to consider the expression which is multiplied to $\beta_3/(\gamma_3+\mu)$. The third term
of this expression also goes to zero as $\delta_2,\delta_3\rightarrow \infty$. As a result, to obtain the limit of $R_0$ as
$\delta_2,\delta_3\rightarrow \infty$ we need to consider the limit of the following expression.
\begin{eqnarray*}
\frac{\beta_3}{\gamma_3+\mu}
\left(
\frac{\delta_2\sigma_1+\delta_1(\alpha_1+\delta_2)}{(\eta+\delta_1)(\alpha_1+\delta_2)}
+\frac{\delta_3\sigma_2}{(\eta+\delta_1)(\alpha_2+\delta_3)} 
\right).
\end{eqnarray*}
As $\delta_2,\delta_3\rightarrow\infty$, the above expression goes to
\begin{eqnarray*}
\frac{\beta_3}{\gamma_3+\mu}
\left(
\frac{\sigma_1+\delta_1}{\eta+\delta_1}
+\frac{\sigma_2}{\eta+\delta_1} 
\right)
= \frac{\beta_3}{\gamma_3+\mu}\cdot \frac{\sigma_1+\sigma_2+\delta_1}{\mu+\sigma_1+\sigma_2+\delta_1}.
\end{eqnarray*}
\end{proof}
Note that $\delta_2$ and $\delta_3$ are the detection rates for the infected-asymptomatic and infected-symptomatic
compartments respectively. As $\delta_2,\delta_3\rightarrow \infty$, the differential equations for $A$ and $I$ are approximated as
$A^{\prime}\approx -\delta_2A$ and $I^{\prime}\approx -\delta_3I$. The solutions are
$A\approx A_0e^{-\delta_2t}$ and $I\approx I_0e^{-\delta_3t}$. For high values of $\delta_2$ and $\delta_3$, the values of $A$ and
$I$ very quickly become close to zero. In other words, if the detection rates $\delta_2$ and $\delta_3$ are high, then within
a short period of time, the system reaches a state where almost all infected individuals are detected.

The limit of $R_0$ as $\delta_2,\delta_3\rightarrow \infty$ given in Theorem~\ref{thm-asymptotic} is
at most $\beta_3/(\gamma_3+\mu)$. The quantity $\beta_3/(\gamma_3+\mu)$ would be the basic reproduction ratio if we
consider only the detected compartment (i.e., leave out the exposed, infected-asymptomatic and infected-symptomatic compartments).
With perfect isolation, $\beta_3$ would be zero and then $R_0$ tends to zero. Even if the isolation is not perfect, it is reasonable
to assume that the basic reproduction ratio for a model consisting only of the detected compartment will be less than one, i.e.,
$\beta_3/(\gamma_3+\mu)<1$. 
So, as $\delta_2,\delta_3\rightarrow \infty$, $R_0$ converges to a limit which is less than one. 
The value of $R_0$ going below one implies that the DFE is asymptotically stable and there is no epidemic. 
So, Theorem~\ref{thm-asymptotic} shows
that if the detection rates of infected-asymptomatic and infected-symptomatic individuals are sufficiently high and
the isolation/recovery process is made effective, then an epidemic can be prevented.  

\section{Conclusion \label{sec-conclu} }
The effects of detection and isolation on controlling an epidemic have been quantifiably demonstrated using several epidemiological
models. Detection rate is determined by the testing methodology followed for the population. A comprehensive and properly targeted
testing strategy will lead to significantly improved detection rates. Similarly, a comprehensive and well designed 
isolation strategy will lead to significantly lower transmission rates of isolated individuals. Also, the recovery rate of isolated
individuals may increase. It has been shown that the joint
effect of these factors can lead to a reduction of the basic reproduction ratio and even make it smaller than $1$. In this context,
we note that both testing and isolation mechanisms are more likely to be successful if they are applied in the early stages
of the spread of the disease when the number of infected persons is comparatively small. If the number of infected persons becomes
high, then it becomes more difficult to detect and isolate a significant number of individuals. 

\section*{Acknowledgement} Thanks to Sanjay Bhattacherjee for comments.


\appendix


\section{Compartment Models \label{sec-comp-models} }
There are many good introductions to compartment models. See for example~\cite{BC18,Li18}. We provide a brief description.

The entire population is divided into disjoint compartments. Suppose there are $n$ compartments out of which $m$ are disease
compartments, with $1\leq m<n$. For $i=1,\ldots,n$, let $X_i$ be the number of persons in compartment $i$. Without loss of 
generality, assume that the first $m$ compartments are disease compartments. The numbers $X_1,\ldots,X_n$ are functions
of time $t$. Let $N$ be the total size of the population so that the invariant $X_1(t)+\cdots+X_n(t)=N$ holds for all $t\geq 0$. 
The dynamics of the system is given by the rate of change of $X_i(t)$, $i=1,\ldots,n$. 
The state of the system at time $t$ is described by the vector $\mathbf{X}(t)=(X_1(t),\ldots,X_n(t))$. 

For $i=1,\ldots,n$, let $\mathcal{F}_i(\mathbf{X}(t)),\mathcal{V}_i^{+}(\mathbf{X}(t))$ and $\mathcal{V}_i^{-}(\mathbf{X}(t))$ be functions
where $\mathcal{F}_i(\mathbf{X}(t))$ is the rate of appearance of new infections in compartment $i$,
$\mathcal{V}_i^{+}(\mathbf{X}(t))$ is the rate of appearance of infections in compartment $i$ by all other means, and
$\mathcal{V}_i^{-}(\mathbf{X}(t))$ is the rate of removal of infections from compartment $i$. Let
$\mathcal{V}_i(\mathbf{X}(t))=\mathcal{V}_i^{-}(\mathbf{X}(t))-\mathcal{V}_i^{+}(\mathbf{X}(t))$.
Suppose the derivatives of $X_1(t),\ldots,X_n(t)$ can be written as
\begin{eqnarray}\label{eqn-DE}
X_i^{\prime}(t) & = & \mathcal{F}_i(\mathbf{X}(t))-\mathcal{V}_i(\mathbf{X}(t)),\quad i=1,\ldots,n,
\end{eqnarray}
Subject to the specification of the initial conditions $X_1(0),\ldots,X_n(0)$, the unfolding of the disease dynamics is described
by the set of differential equations given by~(\ref{eqn-DE}).

Let $f_i(\mathbf{X}(t))=\mathcal{F}_i(\mathbf{X}(t))-\mathcal{V}_i(\mathbf{X}(t))$. The state $\mathbf{X}(t)$ is an equilibrium state
if $f_i(\mathbf{X}(t))=0$, for $i=1,\ldots,n$, i.e., if the rates of change of the sizes of all the compartments are equal to zero. 
An equilibrium state $\mathbf{X}(t)$ is said to be a disease-free equilibrium (DFE) if $X_1(t)=\cdots=X_m(t)=0$, otherwise it is an
endemic equilibrium (EE). A system may have both kinds of equilibrium and also may have more than one equilibrium of each kind.

A parameter of interest is the basic reproduction ratio (or, number) denoted as $R_0$. This parameter determines the stability of the
system near a DFE. If $R_0<1$, then the system is asymptotically stable, 
while if $R_0>1$, then the system is asymptotically unstable (see Theorem~2 of~\cite{vW02}).

The next generation matrix method~\cite{DHJ90,vW02} can be used to calculate the value of $R_0$. Briefly the procedure is the following. 
Suppose $\mathbf{X}_0$ is a DFE. Let $\mathbf{F}$ be an $m\times m$ matrix whose $(i,j)$-th entry is $\partial \mathcal{F}_i/\partial X_j$
evaluated at $\mathbf{X}_0$; let $\mathbf{V}$ be an $m\times m$ matrix whose $(i,j)$-th entry is $\partial \mathcal{V}_i/\partial X_j$
evaluated at $\mathbf{X}_0$. Then $R_0$ is the spectral radius (i.e., the maximum of the absolute values of all the eigen values) of 
$\mathbf{F}\mathbf{V}^{-1}$.


\section{Release Without Confirmation of Being Negative \label{sec-release-wo-test} }

Ideally, individuals who are in the detected compartment would be released after they test negative for the disease. 
WHO guidelines~\cite{WHO-guidelines} recommend that confirmed COVID-19 patients be released only after two negative tests
more than 24 hours apart. The discharge policy~\cite{MHFW2} of the Government of India also followed this recommendation. 
A later document~\cite{MHFW3} of the Government of India, provided a revised discharge policy. This policy mentions that 
mild/very mild/pre-symptomatic and moderate cases may be discharged after 10 days of onset of symptoms and there is no need for 
testing prior to discharge. Three reasons for the change of the discharge policy have been mentioned~\cite{MHFW4}, namely other
countries have adopted similar changes in discharge criteria; laboratory surveillance data indicates that after the initial positive
results, patients become negative after a median duration of 10 days; and recent studies suggest that the viral load peaks in the
pre-symptomatic period (2 days before symptoms) and goes down over the next 7 days. 
From the later two of the above reasons, one cannot be sure that a patient released after 10 days has necessarily become non-infectious. 
Since the median duration is 10 days, there would be patients who remained positive beyond 10 days. Also, the viral load going
down after 2+7 days does not necessarily imply that after 10 days the viral load is such that the individual is non-infectious. 
In fact, the document~\cite{MHFW4} mentions that a patient released without testing is advised to isolate himself/herself at
home for 7 days and to follow certain precautions. 

So, for a patient who is released without confirmation of not being infectious, there remains the possibility of being infectious
for a certain number of days after release. To capture this possibility, we expand the SEAIDR model to include another compartment called 
pseudo-recovered compartment consisting of patients who have been detected to be positive but, have been released from isolation
without being tested to ascertain whether they are negative. This leads to the SEAIDPR model. 
%
Let $P$ be the number of individuals in the pseudo-recovered compartment.
The differential equations describing the dynamics of this model are as follows.
\begin{eqnarray} \label{eqn-SEAIDPR}
\begin{array}{rcl}
S^{\prime} & = & -\frac{\beta_1AS}{N} - \frac{\beta_2IS}{N} - \frac{\beta_3DS}{N} - \frac{\beta_4PS}{N} + \mu N(0) - \mu S, \\
E^{\prime} & = & \frac{\beta_1AS}{N} + \frac{\beta_2IS}{N} +\frac{\beta_3DS}{N} 
		+ \frac{\beta_4PS}{N} - \sigma_1E-\sigma_2E -\delta_1E - \mu E, \\
A^{\prime} & = & \sigma_1E - \delta_2A - \gamma_1A - \kappa A -\mu A, \\
I^{\prime} & = & \sigma_2E + \kappa A - \delta_3I - \gamma_2I - \mu I, \\
D^{\prime} & = & \delta_1E+\delta_2A+\delta_3I-\gamma_3D-\chi D-\mu D, \\
P^{\prime} & = & \chi D-\gamma_4P-\mu P, \\
R^{\prime} & = & \gamma_1A+\gamma_2I+\gamma_3D+\gamma_4P - \mu R.
\end{array}
\end{eqnarray}
In addition to the parameters of the system in~\eqref{eqn-SEAIDR}, the following parameters are required
in~\eqref{eqn-SEAIDPR}.
\begin{itemize}
\item $\beta_4$ and $\gamma_4$ are the transmission and recovery rates respectively of individuals in the pseudo-recovered compartment. 
\item $\chi$ is the transfer rate of individuals from the detected compartment to the pseudo-recovered compartment.
\end{itemize}

The relevant DFE is $(S,E,A,I,D,P,R)=(N,0,0,0,0,0,0)$. The matrices $\mathbf{F}$ and $\mathbf{V}$ of the next generation matrix 
method are as follows.
\begin{eqnarray*}
\mathbf{F} & = &
\left[
\begin{array}{ccccc}
0 & \beta_1 & \beta_2 & \beta_3 & \beta_4 \\
0 & 0 & 0 & 0 & 0 \\
0 & 0 & 0 & 0 & 0 \\
0 & 0 & 0 & 0 & 0 \\
0 & 0 & 0 & 0 & 0 
\end{array}
\right], \\
\mathbf{V} & = &
\left[
\begin{array}{ccccc}
\sigma_1+\sigma_2+\delta_1+\mu & 0 & 0 & 0 & 0 \\
-\sigma_1 & \delta_2+\gamma_1+\kappa+\mu & 0 & 0 & 0 \\
-\sigma_2 & -\kappa & \delta_3+\gamma_2+\mu & 0 & 0 \\
-\delta_1 & -\delta_2 & -\delta_3 & \gamma_3+\chi+\mu & 0 \\
0 & 0 & 0 & -\chi & \gamma_4+\mu \\
\end{array}
\right].
\end{eqnarray*}
The next generation matrix is $\mathbf{F}\mathbf{V}^{-1}$ and the basic reproduction ratio $\widetilde{R}_0$ is the spectral radius of 
the next generation matrix. It is given by
{\small 
\begin{eqnarray}\label{eqn-R0-main-wo-test}
\lefteqn{\widetilde{R}_0} \nonumber \\
& = & 
\frac{\beta_1\sigma_1}{(\eta+\delta_1)(\alpha_1+\delta_2)}
+\frac{\beta_2\sigma_2}{(\eta+\delta_1)(\alpha_2+\delta_3)}
+\frac{\beta_2\sigma_1\kappa}{(\eta+\delta_1)(\alpha_1+\delta_2)(\alpha_2+\delta_3)} \nonumber \\
& & + \frac{\beta_4\chi+\beta_3(\gamma_4+\mu)}{(\gamma_4+\mu)(\gamma_3+\mu+\chi)}
\left(
\frac{\delta_2\sigma_1+\delta_1(\alpha_1+\delta_2)}{(\eta+\delta_1)(\alpha_1+\delta_2)}
+\frac{\delta_3\sigma_2}{(\eta+\delta_1)(\alpha_2+\delta_3)} 
+\frac{\delta_3\sigma_1\kappa}{(\eta+\delta_1)(\alpha_1+\delta_2)(\alpha_2+\delta_3)}
\right). \nonumber \\
\end{eqnarray}
}
The difference between $\widetilde{R}_0$ given by~\eqref{eqn-R0-main-wo-test} and $R_0$ given by~\eqref{eqn-R0-main} is the
appearance of $\chi$, $\beta_4$ and $\gamma_4$ in the expression for $\widetilde{R}_0$. The conditions under which 
$\widetilde{R}_0$ equals $R_0$ are as follows.
\begin{description}
\item {\em $\chi=0$:} Under this condition, there is no transfer into the pseudo-recovered compartment and the 
model collapses to SEAIDR. 
\item{\em $\beta_3=\beta_4$ and $\gamma_3=\gamma_4$:} This condition implies that the transmission and recovery rates of the detected and
the pseudo-recovered compartments are the same. So, again there is no essential difference between these two compartments.
\item{\em $\beta_3=\beta_4=0$:} This means that the transmission rates of individuals in both the detected and pseudo-recovered 
compartments are zero. As a result, neither of these compartments have any effect on $R_0$.
\end{description}
The following result characterises the condition for $\widetilde{R}_0$ to be greater than, equal to, or lesser than $R_0$.
\begin{proposition}\label{eqn-prop-wo-test}
Suppose $\chi>0$. Then $\widetilde{R}_0$ (given by~\eqref{eqn-R0-main-wo-test}) is greater than, equal to, or lesser than $R_0$ 
(given by~\eqref{eqn-R0-main}) according as $\beta_4/(\gamma_4+\mu)$ is greater than, equal to, or lesser than $\beta_3/(\gamma_3+\mu)$.
\end{proposition}
\begin{proof}
From the expressions given in~\eqref{eqn-R0-main-wo-test} and~\eqref{eqn-R0-main}, we have
$\widetilde{R}_0$ is greater than, equal to, or lesser than $R_0$ according as
$$
\frac{\beta_4\chi+\beta_3(\gamma_4+\mu)}{(\gamma_4+\mu)(\gamma_3+\mu+\chi)}
$$
is greater than, equal to, or lesser than
$$
\frac{\beta_3}{\gamma_3+\mu}.
$$
Simplification of this condition provides the condition stated in the proposition.
\end{proof}

From Proposition~\ref{eqn-prop-wo-test}, it follows that if $\beta_4/(\gamma_4+\mu)$ is greater than $\beta_3/(\gamma_3+\mu)$,
then $\widetilde{R}_0$ is greater than $R_0$, i.e., the basic reproduction ratio actually increases. Since patients in the pseudo-recovered 
compartment have not been confirmed by a test to have become non-infectious and are also not under strict isolation, the possibility of the 
condition $\beta_4/(\gamma_4+\mu)>\beta_3/(\gamma_3+\mu)$ cannot be completely ruled out. If this condition indeed happens to hold, then 
it causes an adverse effect on the control of the epidemic that might otherwise have been attained through effective testing
and isolation. 

In a manner similar to the proof of Theorem~\ref{thm-asymptotic}, we can prove the following result.
\begin{theorem}\label{thm-asym-wo-test}
For fixed values of $\beta_1,\beta_2,\beta_3,\beta_4,\gamma_1,\gamma_2,\gamma_3,\gamma_4,\sigma_1,\sigma_2,\kappa,\chi$ and $\delta_1$,
$$\widetilde{R}_0\rightarrow 
\frac{\beta_4\chi+\beta_3(\gamma_4+\mu)}{(\gamma_4+\mu)(\gamma_3+\mu+\chi)}
\cdot \frac{\sigma_1+\sigma_2+\delta_1}{\mu+\sigma_1+\sigma_2+\delta_1}, \quad \mbox{ as }
\delta_2,\delta_3\rightarrow \infty. $$
\end{theorem}
Let us assume that $\frac{\sigma_1+\sigma_2+\delta_1}{\mu+\sigma_1+\sigma_2+\delta_1} \approx 1$ so that the limit of
$\widetilde{R}_0$ given by Theorem~\ref{thm-asym-wo-test} is approximately
\begin{eqnarray}\label{eqn-t1a}
\rho & = & \frac{\beta_4\chi+\beta_3(\gamma_4+\mu)}{(\gamma_4+\mu)(\gamma_3+\mu+\chi)}.
\end{eqnarray}
Suppose $\beta_3/(\gamma_3+\mu)<1$. Then by the discussion after Theorem~\ref{thm-asymptotic}, if $\chi=0$ (i.e., there is no
release from isolation without confirmation of being non-infectious), there is no epidemic. Now, assume that $\beta_4/(\gamma_4+\mu)>1$,
i.e., control in the pseudo-recovered compartment is ineffective.
Then $\rho>1$ if and only if 
\begin{eqnarray}\label{eqn-t2a}
\chi & > & \frac{(\gamma_4+\mu)(\gamma_3+\mu-\beta_3)}{\beta_4-\gamma_4-\mu}.
\end{eqnarray}
So, if the transfer rate from detected to pseudo-recovered compartment is higher than the threshold given by~\eqref{eqn-t2a}, then
the approximate value of $R_0$ given by $\rho$ in~\eqref{eqn-t1a} is greater than 1. In other words, if a lot of individuals get
transferred from detected to pseudo-recovered compartment and the control in the pseudo-recovered compartment is not effective, then
there will be an epidemic which might otherwise have been prevented without such transfer taking place. 

\section{SAGE Code \label{sec-SAGE-code} }
The SAGE code to compute $R_0$ for the SIDR model is the following.
{\small
\begin{verbatim}
reset()
var('b0','b1','d','g0','g1','m')
F = Matrix( [ [b0, b1], [0, 0] ] )
V = Matrix( [ [d+g0+m, 0], [-d, g1+m] ] )
M = F*V^(-1)
lst = M.eigenvalues()
t1 = numerator(lst[0])
t2 = factor(denominator(lst[0]))
print factor(t2)
print t1
\end{verbatim}
}
The SAGE code to compute $R_0$ for the SEAIR model is the following.
{\small
\begin{verbatim}
reset()
var('b1','b2','b3','s1','s2','d1','d2','d3','g1','g2','g3','k','m')
F = Matrix( [ [0, b1, b2], [0,0,0], [0,0,0] ] )
V = Matrix( [ [s1+s2+m, 0, 0], [-s1, g1+k+m, 0], [-s2, -k, g2+m] ] )
M = F*V^(-1)
lst = M.eigenvalues()
t1 = numerator(lst[0])
t2 = denominator(lst[0])
print factor(t2)
print t1
\end{verbatim}
}
The SAGE code to compute $R_0$ for the SEAIDR model is the following.
{\small
\begin{verbatim}
reset()
var('b1','b2','b3','s1','s2','d1','d2','d3','g1','g2','g3','k','m','d')
F = Matrix( [ [0, b1, b2, b3], [0,0,0,0], [0,0,0,0], [0,0,0,0] ] )
V = Matrix( [ [s1+s2+d1+m, 0, 0, 0], [-s1, d2+g1+k+m, 0, 0], 
[-s2, -k, d3+g2+m, 0], [-d1, -d2, -d3, g3+m] ] )
M = F*V^(-1)
lst = M.eigenvalues()
t1 = numerator(lst[0])
t2 = denominator(lst[0])
print factor(t2)
print t1
\end{verbatim}
}
The SAGE code to compute $R_0$ for the SEAIDPR model is the following.
{\small
\begin{verbatim}
reset()
var('b1','b2','b3','b4','s1','s2','d1','d2','d3','g1',
    'g2','g3','g4','k','m','d','x')
F = Matrix( [ [0, b1, b2, b3, b4], [0,0,0,0,0], [0,0,0,0,0], 
[0,0,0,0,0], [0,0,0,0,0] ] )
V = Matrix( [ [s1+s2+d1+m, 0, 0, 0, 0], [-s1, d2+g1+k+m, 0, 0, 0], 
[-s2, -k, d3+g2+m, 0, 0], [-d1, -d2, -d3, g3+x+m, 0], [0, 0, 0, -x, g4+m] ] )
M = F*V^(-1)
lst = M.eigenvalues()
t1 = numerator(lst[0])
t2 = denominator(lst[0])
print factor(t2)
print t1
\end{verbatim}
}

\end{document}